\documentclass[prodmode,acmtkdd]{acmsmall} 

\usepackage[ruled]{algorithm2e}
\usepackage{amsmath,amssymb}
\usepackage{graphicx}
\DeclareMathOperator*{\argmax}{arg\,max}
\renewcommand{\algorithmcfname}{ALGORITHM}
\SetAlFnt{\small}
\SetAlCapFnt{\small}
\SetAlCapNameFnt{\small}
\SetAlCapHSkip{0pt}
\IncMargin{-\parindent}

\acmVolume{0}
\acmNumber{0}
\acmArticle{0}
\acmYear{0}
\acmMonth{0}

\begin{document}

\markboth{Z. Wang et al.}{Information Coverage Maximization in Social Networks}

\title{Information Coverage Maximization in Social Networks}
\author{ZHEFENG WANG
\affil{University of Science and Technology of China}
ENHONG CHEN
\affil{University of Science and Technology of China}
QI LIU
\affil{University of Science and Technology of China}
YU YANG
\affil{Simon Fraser University}
YONG GE
\affil{¡ì University of North Carolina at Charlotte}
BIAO CHANG
\affil{University of Science and Technology of China}
}

\begin{abstract}
Social networks, due to their popularity, have been studied extensively these years. A rich body of these studies is related to influence maximization, which aims to select a set of seed nodes for maximizing the expected number of active nodes at the end of the process. However, the set of active nodes can not fully represent the true coverage of information propagation. A node may be informed of the information when any of its neighbours become active and try to activate it, though this node (namely informed node) is still inactive. Therefore, we need to consider both active nodes and informed nodes that are aware of the information when we study the coverage of information propagation in a network. Along this line, in this paper we propose a new problem called \textsl{Information Coverage Maximization} that aims to maximize the expected number of both active nodes and informed ones. After we prove that this problem is NP-hard and submodular in the independent cascade model and the linear threshold model, we design two algorithms to solve it. Extensive experiments on three real-world data sets demonstrate the performance of the proposed algorithms.
\end{abstract}





\begin{bottomstuff}
This is a technical report.
\end{bottomstuff}

\maketitle

\section{Introduction}
Recent years have witnessed the popularity of online social networking sites such as Facebook and Twitter. Many people spend much time on these sites and share different kinds of information with their friends. Social networks play important roles in the spread of information, ideas or opinions. Therefore, the analysis of information propagation in social networks has been a critical research area these years.

In the literature, many efforts have been made on the development of information propagation models. For example, \textsl{Independent Cascade}~(IC)~model~\cite{goldenberg2001talk} and \textsl{Linear Threshold}~(LT)~model~\cite{granovetter1978threshold}, a data-based credit distribution model~\cite{goyal2011data} and linear social influence model~\cite{ijcai13biaoxiang} were proposed to describe the information diffusion process. Among these models, IC and LT models are stochastic diffusion models~\cite{chen2013information} which specify the randomized process of information propagation. In these models, each node in the network has two possible states: active and inactive. Intuitively, an active node can be viewed as adopting the new information that is propagated in the network. During the diffusion process, the active nodes will try to activate their neighbors and the inactive nodes will not.

Given an information propagation model, most of the existing works focused on selecting a set of seed nodes to be activated that could lead to the maximum expected number of active nodes. This selection problem is formulated as a discrete optimization problem called \textsl{Influence Maximization}~\cite{kempe2003maximizing}. This problem, due to its important application in viral marketing, has been extensively explored (~\cite{kimura2006tractable,Wang_community,liu2010mining,kim2013scalable,borgs2014maximizing,wang2014influential}~).
\begin{figure}
  \centering
  \includegraphics[width=0.75\textwidth]{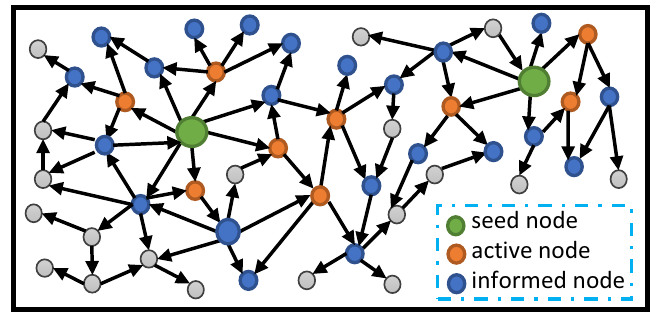}\\
  \caption{Information propagation in a social network}\label{fig:social_graph}
\end{figure}

However, during the process of information propagation, there are actually two types of inactive nodes. For example, when we publish a message in Twitter, some users may retweet the message and others may not. But, among all users who have not retweeted the message, many of them may be aware of this message as their friends have retweeted it, while the rest is truly inactive. An example of such information propagation in a social network is shown in Figure~\ref{fig:social_graph}. If we take a close look at the process of information propagation in this example, we will find that a node may be informed of the information if at least one of its neighbours become active. We call such nodes as \textsl{informed nodes} in this paper. In contrast, a node may never know the information if none of its neighbours is active. In fact, there are a large number of informed nodes in many real-world social networks as we will show in our experiment later. Influence maximization only considers the active nodes and neglects the informed nodes, thus it can not model the true coverage of information propagation well. To better measure the coverage of information propagation, we should consider both active nodes and informed nodes.

To this end, we formulate a new problem called \textsl{Information Coverage Maximization} to address this issue. The objective of this problem is to maximize the expected number of both active nodes and informed nodes. We prove that the problem is NP-hard and submodular in the IC and LT model. We also show that computing exact information coverage in the IC model and LT model is \#P-hard. Then, we design two algorithms to solve the proposed problem. Finally, we evaluate the proposed algorithms with three real-world data sets. The experimental results demonstrate the performance of the proposed algorithms. Our contributions can be summarized as follows:
\begin{itemize}
\item We distinguish the informed node from the inactive node, and explore the value of informed nodes to better measure the coverage of information propagation. Thus, we propose a new problem of maximizing the expected number of both active nodes and informed nodes.
\item We prove that the proposed problem is NP-hard and the computation of information coverage is \#P-hard in the IC model and LT model. We also show that the objective function is submodular in the IC model and LT model.
\item We design two algorithms to solve the proposed problem. The proposed algorithms are examined with three real-world data sets and the experimental results show the performance of the proposed algorithms.
\end{itemize}

\textbf{Overview.} The rest of the paper is organized as follows. In Section $2$, we discuss related works. Section $3$ gives the definition of the problem and shows the properties of the problem. In Section $4$, we design three algorithms to solve the proposed problem. Section $5$ presents the experimental results. In Section $6$, we conclude our work.
\section{Related Work}
Social networks have been studied extensively for many years. A rich body of these studies is focused on the analysis of influence and information propagation in social networks. Several models have been proposed to describe the diffusion of information through the social network, such as IC model~\cite{goldenberg2001talk}, LT model~\cite{granovetter1978threshold} and decreasing cascade model~\cite{kempe2005influential}. These models define the stochastic process of information propagation. Thus they are called stochastic diffusion models~\cite{chen2013information}. There are also models which formulate the information propagation from other perspectives (~\cite{aggarwal2011flow,goyal2011data,ijcai13biaoxiang}~). Moreover, in~\cite{chen2012time} and~\cite{liu2012time}, the authors extended IC model to consider the time-delay aspect of influence diffusion.

Influence maximization~\cite{kempe2003maximizing}, which aims to maximize the expected number of active nodes in a given diffusion model, is another main research direction of the analysis of information propagation in social networks. In~\cite{kempe2003maximizing}, the authors proved the problem is NP-hard in both IC and LT models and proposed a greedy framework to solve it. The following researchers focused on developing both efficient and effective algorithms, such as CELF~\cite{leskovec2007cost}, PMIA~\cite{chen2010scalable}, LDAG~\cite{chen2010scalablelt}, SIMPATH~\cite{goyal2011simpath}, StaticGreedy~\cite{cheng2013staticgreedy}, Linear and Bound~\cite{Liu_2014_im} and IMRank~\cite{Cheng_imrank}. In addition, in~\cite{chen2012time} and~\cite{liu2012time}, the authors studied the influence maximization with time-critical constraint. In~\cite{tang2014diversified}, the authors studied the diversified influence maximization which considers both the magnitude of influence and the diversity of the influenced crowd. But influence maximization only considers the active nodes, which makes it different from the proposed problem.
\section{Problem Formulation}
In this section, we first give a formal definition of the information coverage maximization problem. Then we discuss the computational complexity of the proposed problem. Finally, we show some properties of the objective function.
\subsection{Problem Definition}
Let the directed graph$G=(V,E)$ denote an information propagation network, where $V=\{1,2,...n\}$ is the set of nodes and $E$ is the set of directed edges between nodes. A node in the graph corresponds to an individual in the social network and the directed edges represent the relationships between the individuals. In this paper, we use $n$ to denote the number of nodes and $m$ to denote the number of edges respectively.

Although there are quite a few diffusion models available to describe the process of information diffusion, we focus on the two most widely used models: IC model and LT model in this paper. In the IC model, there is a propagation probability matrix $P=[p_{i,j}]_{n*n}$ to denote the probability of node $i$ on activating node $j$. In the LT model, there is a propagation weight matrix $Q=[q_{i,j}]_{n*n}$ to denote the importance of node $i$ on activating node $j$.

In both IC model and LT model, seed nodes are the initial active nodes selected to propagate the information and they will try to activate their neighbours. Their neighbours will be informed of the information and may be activated. If a node is activated, it becomes an active node and will try to activate its own neighbours. If a node is not activated but receives the information, then it is an informed node. The process continues until no more nodes can be activated.

Let $S$, $A$, and $L$ denote the seed nodes, active nodes and informed nodes respectively. Then we get the relationships between them as follows:
\begin{equation}\label{eq:relation}
\begin{aligned}
A&=I(S) \\
L&=\bigcup_{a\in A}N(a)
\end{aligned}
\end{equation}
Where $I(S)$ is the set of final active nodes when the information diffusion process converges and $N(a)$ is the set of inactive out neighbours of node $a$.

Then, we can define the information coverage as follows:
\begin{definition} \label{def_ic}
\textbf{Information Coverage.} Given an information propagation network $G=(V,E)$, an information diffusion model on $G$, and a seed set $S$, the information coverage is the sum of expected number of active nodes and informed nodes.
\begin{equation}\label{eq:ic}
F(S)=E(|A|)+E(|L|) \\
\end{equation}
\end{definition}
Considering the relationship given by Eq.~(\ref{eq:relation}), we can rewrite Eq.~(\ref{eq:ic}) as follows:
\begin{equation}\label{eq:ic_rewrite}
F(S)=E(|I(S)|)+E(|\bigcup_{a\in I(S)}N(a)|) \\
\end{equation}
Now, we can give a formal definition of the information coverage maximization problem as follows:
\begin{definition} \label{def_icm}
\textbf{Information Coverage  Maximization.} Given an information propagation network $G=(V,E)$, an information diffusion model on $G$, and a budget number $k$, find a seed set $S$ with $|S|=k$ such that the information coverage $F(S)$ under the given diffusion model is maximized.
\begin{equation}\label{eq:icm}
S^*=\argmax_{|S|=k}~F(S) \\
\end{equation}
\end{definition}

Comparing the objective function $F(S)$ to the one of traditional influence maximization problem, we can see that the first term of $F(S)$ is exactly the influence spread~\cite{kempe2003maximizing}. The difference is that $F(S)$ contains the expected number of informed nodes, which makes it better model the true range of information propagation.

In the real world, the informed nodes may have different values than the active nodes. Therefore, we introduce a weight coefficient to control the relative values of the informed nodes. Then we can define the \textbf{Weighted Information Coverage} as follows:
\begin{definition}\label{def_wc}
\textbf{Weighted Information Coverage.} Given an information propagation network $G=(V,E)$, an information diffusion model on $G$, and a seed set $S$, the weighted information coverage is the weighted sum of expected number of active nodes and informed nodes.
\begin{equation}\label{eq:wc}
\begin{aligned}
W(S)&=E(|A|)+\lambda~E(|L|) \\
s.t.\quad& \lambda \in [0,1]
\end{aligned}
\end{equation}
\end{definition}
 The weight coefficient $\lambda$ controls the importance of informed nodes. When $\lambda$ equals to $1$, $W(S)$ equals to the information coverage $F(S)$. When $\lambda$ equals to $0$, $W(S)$ is the same as the influence spread. Thus, both the information coverage and the influence spread are special cases of the weighted information coverage. To this end, we can define a general form of information coverage maximization problem as follows:
\begin{definition} \label{def_wcm}
\textbf{Weighted Information Coverage Maximization.} Given an information propagation network $G=(V,E)$, an information diffusion model on $G$, and a budget number $k$, find a seed set $S$ with $|S|=k$ such that the weighted information coverage $W(S)$ under the given diffusion model is maximized.
\begin{equation}\label{eq:wc}
S^*=\argmax_{|S|=k}~W(S)
\end{equation}
\end{definition}

 \subsection{Computational Complexity}
In this part, we discuss the computational complexity of the proposed problems in IC model and LT model respectively.
\begin{theorem}
Both the information coverage maximization problem and the weighted information coverage maximization problem are NP-hard in the IC model.
\end{theorem}
\begin{proof}
We reduce from the set cover problem~\cite{karp1972reducibility} to prove this theorem. The definition of the set cover problem is: given a collection of subsets $S_1, S_2,..., S_m$ of a ground set $U = \{u_1, u_2,..., u_n\}$, the question is if there exist $k$ of the subsets whose union is $U$.

Given an arbitrary instance of the set cover problem, we construct a corresponding directed bipartite graph: there is a node $i$ for each subset $S_i$, a node $j$ for each element $u_j$, and a directed edge$(i,j)$ with a propagation probability $p_{i,j}=0$ when $u_j\in S_i$. Since all probabilities are $0$, the information propagation is a deterministic process in this case. Thus, the set cover problem is equivalent to deciding if there is a set $N$ of $k$ nodes in the graph with $F(N)=n+k$. If any set $N$ of $k$ nodes has $F(N)=n+k$, then we can initially activate the $k$ nodes corresponding to subsets such that all $n$ nodes corresponding to elements in the ground set will be informed. This means that the set cover problem must be solvable. For the weighted case, the set cover problem is equivalent to deciding if there is a set $N$ of $k$ nodes in the graph with $W(N)=\lambda~n+k$.
\end{proof}

\begin{theorem}
Given a seed set $S$, computing the information coverage $F(S)$ or the weighted information coverage $W(S)$ is \#P-hard in the IC model.
\end{theorem}
\begin{proof}
We reduce from the $s-t$ connectedness problem~\cite{valiant1979complexity} to prove the theorem. The definition of the $s-t$ connectedness problem is: given a directed graph $G=(V,E)$ and two nodes $s$ and $t$ in the graph, the question is to count the number of subgraphs of $G$ in which $s$ is connected to $t$. In~\cite{chen2010scalable}, the authors show that this problem is equivalent to computing the probability that $s$ is connected to $t$ when each edge in $G$ is connected with a probability of $\frac{1}{2}$.

Given an arbitrary instance of the $s-t$ connectedness problem, let $W_G(S)$ and $F_G(S)$ denote the (weighted) information coverage of seed set $S$ in graph $G$ respectively. Then let $S=\{s\}$ and $p(e)=\frac{1}{2}$ for all $e \in E$, and compute $I_1=F_G(S)$. Next, add a new node $t'$ and a directed edge from $t$ to $t'$ with a propagation probability $p_{t,t'}=1$. Now we obtain a new graph $G'$ and compute $I_2=F_{G'}(S)$. Let $p_G(S,t)$ denote the probability that node $t$ is activated by $S$. Since graph $G'$ only has an extra node $t'$, it is easy to see that $I_2=F_G(S)+p_G(S,t)(p_{t,t'}+1-p_{t,t'})$. Thus, $I_2-I_1$ is the probability that $s$ is connected to $t$. This means that $s-t$ connectedness problem must be solvable. For the $W(S)$ case, $I_1=W_G(S)$ and $I_2=W_G(S)+p_G(S,t)(p_{t,t'}+\lambda(1-p_{t,t'}))$.
\end{proof}

\begin{theorem}
Both the information coverage maximization problem and the weighted information coverage maximization problem are NP-hard in the LT model.
\end{theorem}
\begin{proof}
We reduce from the vertex cover problem~\cite{karp1972reducibility} to prove this theorem. The definition of the vertex cover problem is: given a graph $G=(V,E)$ and a positive integer $k$, the question is if there is vertex set of size $k$ such that there is at least one endpoint in this set for each edge in the graph.

Given an arbitrary instance of the vertex cover problem, we construct a new graph $G'$ like this: First, for each edge $(u,v)$ in graph $G$, we associate it with a propagation weight $q_(u,v)=1/degree(v)$. Second, for each vertex $v$ in graph $G$, we add a new vertex $v'$ and a directed edge from $v$ to $v'$ with a propagation weight $q_{v,v'}=0$. Then the vertex cover problem is equivalent to deciding if there is a node set $N$ of size $k$ such that $F(N)=2n$ (assuming the number of vertices in graph $G$ is $n$). if there is any node set $N$ of size $k$ has $F(N)=2n$, then the node set $N$ is a vertex cover of size $k$ of the graph $G$. This means that the vertex cover problem is solvable. For the weighted case, the vertex cover problem is equivalent to deciding if there is a node set $N$ of size $k$ such that $W(N)=(1+\lambda)n$.
\end{proof}
\begin{theorem}
Given a seed set $S$, computing the information coverage $F(S)$ or the weighted information coverage $W(S)$ is \#P-hard in the LT model.
\end{theorem}
\begin{proof}
We reduce from the influence spread computation problem to prove the theorem. In ~\cite{chen2010scalablelt}, the authors proved computing influence spread in the LT model is \#P-hard.

Given an arbitrary instance of the influence spread computation problem, let $x$ and $y$ denote the expected number of active nodes and informed nodes respectively. $x$ is exactly the influence spread in the graph and the weighted information coverage is $W(S)=x+\lambda*y$. Then for each node $v$ in the graph, we add a new node $v'$ and a directed edge from $v$ to $v'$ with a propagation weight $q_(v,v')=0$. Now we obtain a new graph $G'$. Since the propagation weight of new edge is $0$, the expected number of active nodes in the new graph is still $x$. Thus the weighted information coverage in the new graph is $W'(S)=x+\lambda*(x+y)$. Now, we can get $x=\frac{W'(S)-W(S)}{\lambda}$. This means that the influence spread computation problem is solvable. For the F(S) case, we can get $x=W'(S)-W(S)$.
\end{proof}
In the above proof, we assumed that $\lambda$ is a predefined constant. If we view $\lambda$ as an input of the weighted information coverage $W(S)$, we will have a stronger result.
\begin{theorem}
If $\lambda$ is an input of the weighted information coverage $W(S)$, computing $W(S)$ is \#P-hard whenever the computation of influence spread is \#P-hard.
\end{theorem}
\begin{proof}
Given an information propagation network $G$, and a diffusion model on $G$, let let $x$ and $y$ denote the expected number of active nodes and informed nodes respectively. $x$ is exactly the influence spread in the graph and the weighted information coverage is $W(S)=x+\lambda*y$. Since $\lambda$ is an input of the weighted information coverage $W(S)$, we can change the value of $lambda$ and compute the $W(S)$ multiple times. For example, we can get $W_1=x+\lambda_1*y$ and $W_2=x+\lambda_2*y$. Then we can solve $x$ from the two equations. It follows the result of the theorem.
\end{proof}
\subsection{The Properties of Objective Functions}\label{subsec:func_pro}
In this part, we show that the objective functions $F(\cdot)$ and $W(\cdot)$ have the following properties:
\begin{itemize}
\item $F(\emptyset)=0$ and $W(\emptyset)=0$.
\item Both $F(\cdot)$ and $W(\cdot)$ are monotone.
\item Both $F(\cdot)$ and $W(\cdot)$ are submodular.
\end{itemize}
Since the first two properties are straightforward, we focus on proving the third one.
\begin{theorem}
Both $F(\cdot)$ and $W(\cdot)$ are submodular in the IC model and LT model.
\end{theorem}
\begin{proof}
We utilize the live-arc graph model~\cite{kempe2003maximizing} to prove the theorem. Given an information propagation graph $G$, we construct the live-arc graphs for the IC model and LT model respectively. Then the following proof is the same for the two models. Let $G_{L}$ denote a random live-arc graph, and let $Prob(G_{L})$ denote the probability that $G_{L}$ is selected from all possible live-arc graphs. Let $R_{G_{L}}(S)$ denote the set of all nodes that can be reached from $S$ in $G_{L}$. Then $R_{G_{L}}(S)$ is exactly the active nodes when $S$ is the seed nodes. Next, let $U_{G_{L}}(S)$ denote the union of the inactive out neighbours of the active nodes. Now, for both the IC model and LT model, we have
\begin{equation}\label{eq:live_arc}
\begin{aligned}
&C_{G_L}(S)=|R_{G_{L}}(S)|+\lambda*|U_{G_L}(S)| \\
&W(S)=\sum_{all ~possible ~G_{L}}{Prob(G_{L})C_{G_L}(S)}
\end{aligned}
\end{equation}
Since an non-negative linear combination of submodular functions is also submodular, we only need to prove $C_{G_{L}}(\cdot)$ is submodular for any live-arc graph $G_L$. To do this, Let $M$ and $N$ be two sets of nodes such that $M \subseteq N \subseteq V$ and $v\in V\setminus~N$. Then we have
\begin{equation}\label{eq:submodular}
\begin{aligned}
C_{G_{L}}(M\cup {v})-C_{G_{L}}(M)=~~~~&|R_{G_{L}}(v)|+\lambda*|U_{G_L(v)}|-|R_{G_{L}}(v)\cap R_{G_{L}}(M)|\\
-&\lambda*|R_{G_{L}}(v)\cap U_{G_{L}}(M)|-\lambda*|U_{G_L(v)}\cap U_{G_L(M)}|
\end{aligned}
\end{equation}
\begin{equation}\label{eq:submodular2}
\begin{aligned}
C_{G_{L}}(N\cup {v})-C_{G_{L}}(N)=~~~~&|R_{G_{L}}(v)|+\lambda*|U_{G_L(v)}|-|R_{G_{L}}(v)\cap R_{G_{L}}(N)|\\
-&\lambda*|R_{G_{L}}(v)\cap U_{G_{L}}(N)|-\lambda*|U_{G_L(v)}\cap U_{G_L(N)}|
\end{aligned}
\end{equation}
Since we have $M \subseteq N$, then we can get $C_{G_{L}}(M\cup {v})-C_{G_{L}}(M) \geq C_{G_{L}}(N\cup {v})-C_{G_{L}}(N)$. It follows that $C_{G_L}(\cdot)$ is submodular. Thus $W(\cdot)$ is submodular. For the $F(\cdot)$ case, let $\lambda=1$ and the result still holds.
\end{proof}
\section{Solutions}
We have shown the computational complexity of the proposed problems in the previous section. Thus we can not find the optimal solution or compute the exact information coverage in polynomial time under the assumption $P \neq NP$. In this section, we discuss an approximation algorithm and two heuristic algorithms.
\subsection{Greedy Algorithm with Lazy Evaluation Optimization}
In Section~\ref{subsec:func_pro}, we show that $F(\cdot)$ and $W(\cdot)$ have three properties. Based on these properties, we can design a simple greedy strategy: add the node that provides the largest marginal contribution to the objective function in each iteration. According to~\cite{nemhauser1978analysis}, the greedy strategy can approximate the optimal solution with a factor of $1-\frac{1}{e}$. However, the greedy strategy relies on the exact computation of the objective function. In our case, computing the objective function is \#P-hard. Thus we need to use Monte Carlo simulation method to estimate the objective function. Then as shown in~\cite{chen2013information}, the greedy strategy with Monte Carlo simulation has an approximation ratio of $1-\frac{1}{e}-\epsilon$, where $\epsilon$ is a constant number dependent on the accuracy of the Monte Carlo simulation. In order to get a good approximation, we have to run  Monte Carlo simulations for sufficiently many times (e.g., 10,000). Consequently, the greedy strategy is very time-consuming. Due to the submodularity of the objective function, we adopt a optimization trick called lazy evaluation~\cite{minoux1978accelerated} to speed up the greedy strategy. Let $\Delta_M(v)=F(M \cup v)-F(M)$ denote the marginal gain after adding $v$ to $M$. Then for $M \subseteq N \subseteq V$, we have $\Delta_M(v) \geq \Delta_N(v)$. Thus we can use the marginal gain computed in the previous iteration as a upper bound of the current iteration. We only update the marginal gain when necessary. In this way, the lazy forward update scheme can effectively reduce the number of the objective function evaluations. More details about the update scheme are shown in Algorithm~\ref{alg_lfg}. From the algorithm, we can see that it needs $(n+k\beta)$ times of objective function evaluations, where $\beta \ll n$ is the expected number of objective function evaluations in each iteration. Thus the average time complexity is $O(nRm+k\beta Rm)$, where $R$ is the number of rounds of simulations in each estimation.

\begin{algorithm}[h]
\SetAlgoNoLine
\caption{The Lazy-Forward Greedy Algorithm} \label{alg_lfg}
\KwIn{$G=(V,E,T)$, number $k$}
\KwOut{seed set $S$}
initialize $S=\emptyset$ \\
\For {each node $n$ in $V$}
{
    //for the weighted case, replace $F(\cdot)$ with $W(\cdot)$ \\
    compute $\Delta(n)=F(n)$ \\
    $stamp_{n}=0$
}
\While {$|S|<k$}
{

    $n=\argmax_{n\in V\setminus~S}{\Delta(n)}$   \\
    \If{$stamp_{n}==|S|$}
    {
        $S=S\cup~{n}$  \\
    }\Else
    {
     //for the weighted case, replace $F(\cdot)$ with $W(\cdot)$ \\
        compute $\Delta(n)=F(S\cup~{n})-F(S)$ \\
        $stamp_{n}=|S|$
    }
}
\Return $S$
\end{algorithm}

\subsection{Degree Based Heuristic Algorithm}
To address the scalability issue, we develop an efficient degree based heuristic algorithm. When we revisit the objective function, we can find that a node's contribution to the objective function is highly dependent on its out degree. Thus if we rank the nodes according to their out degrees and take top-k nodes as the seed nodes, we can probably get a good result. Furthermore, when a node is selected, its out neighbours will be informed. This will result in a decrease of other nodes' ``effective'' out degrees, as their out neighbours may have been informed. This observation means that we can benefit from adjusting each node's ``effective'' out degree dynamically. This heuristic is summarized in Algorithm~\ref{alg_edr}. From the algorithm, we can see that it takes only $O(k(n+m))$ time to complete if we store the graph $G$ and the covered nodes set $C$ with appropriate data structures.

\begin{algorithm}[h]
\SetAlgoNoLine
\caption{The Effective Degree Rank Algorithm} \label{alg_edr}
\KwIn{$G=(V,E,T)$, number $k$}
\KwOut{seed set $S$}
initialize $S=\emptyset$ \\
initialize $C=\emptyset$ \\

\For{each node $n$ in $V$}
{
    $EffectiveDegree(n)=OutDegree(n)$
}
\While {$|S|<k$}
{
    $n=\argmax_{n\in V\backslash S}{EffectiveDegree(n)}$   \\
    $S=S\cup~{n}$\\
    $C=C\cup~OutNeighbour(n)$ \\

    \For{each node $n$ in $V\setminus~S$}
    {
        $EffectiveDegree(n)=OutDegree(n)-|C\cap~OutNeighbour(n)|$ \\
    }

}
\Return $S$

\end{algorithm}

\section{Conclusion}
In this paper, to better measure the coverage of information propagation, we distinguish the informed node from the inactive node and explore the value of the informed nodes. Meanwhile, we formulate a novel problem called information coverage maximization which aims to maximize the expected number of both active nodes and informed nodes. Furthermore, we prove the proposed problem is NP-hard and submodular in the IC model and LT model. We also show that the computation of information coverage is \#P-hard in IC model and LT model. Then based on the properties of the problem, we design two algorithms to solve it. Finally, we conduct extensive experiments to verify our idea. The experimental results show the difference between influence maximization and information coverage maximization. The performance of the proposed algorithms is also demonstrated in the experiments. We hope our study could lead to more future works.


\begin{acks}

\end{acks}

\bibliographystyle{ACM-Reference-Format-Journals}
\bibliography{icm}



\medskip

\end{document}